\documentclass{article}
\usepackage[utf8]{inputenc}
\usepackage[]{amsmath}
\usepackage{amsthm}
\usepackage[margin=40mm, top=35mm, bottom=35mm, footnotesep=10mm]{geometry}
\usepackage{amssymb, amsfonts}
\usepackage[english]{babel} 
\usepackage{float}
\usepackage{subfigure}
\usepackage{graphicx}
\usepackage{color}
\usepackage{tabularx}
\usepackage{array}
\usepackage{natbib}

\newcolumntype{C}[1]{>{\centering\arraybackslash}p{#1}}

\newtheorem{thm}{Theorem}[section]
\newtheorem{rmk}{Remark} [section]
\newtheorem{ass}{Assumption} [section]
\allowdisplaybreaks

\newcommand{\p}{p_{cont}}
\newcommand{\pistar}{\pi^*}
\newcommand{\Pistar}{\Pi^*}
\newcommand{\astar}{a^*}

\newcommand{\jump}{\leavevmode\hphantom{0}}
\newcommand{\jumpminus}{\leavevmode\hphantom{-}}

\begin{document}
	
	\title{General Semi-Markov Model for Limit Order Books:\\
		Theory, Implementation and Numerics}

	\author{Anatoliy Swishchuk\thanks{Department of Mathematics and Statistics, University of Calgary, 2500 University Drive NW,Calgary, Alberta T2N 1N4, Canada, aswish@uclagary.ca},  Katharina Cera\thanks{Department of Mathematics, Technical University of Munich, Boltzmannstr. 3, 85747 Garching, Germany and Department of Mathematics and Statistics, University of Calgary, 2500 University Drive NW, Calgary, Alberta T2N 1N4, Canada, Katharina.Cera@mytum.de},  Julia Schmidt\thanks{Department of Mathematics, Technical University of Munich, Boltzmannstr. 3, 85747 Garching, Germany and Department of Mathematics and Statistics, University of Calgary, 2500 University Drive NW, Calgary, Alberta T2N 1N4, Canada, Julia.Schmidt@tum.de}  \text{ }and Tyler Hofmeister\thanks{Department of Mathematics and Statistics, University of Calgary, 2500 University Drive NW,Calgary, Alberta T2N 1N4, Canada, Tyler.Hofmeister@ucalgary.ca} }
	
	\date{August 2016}
	
	\maketitle
	
	\begin{abstract}
		The paper considers a general semi-Markov model for Limit Order Books with two states, which incorporates price changes that are not fixed to one tick. Furthermore, we introduce an even more general case of the semi-Markov model for Limit Order Books that incorporates an arbitrary number of states for the price changes. For both cases the justifications, diffusion limits, implementations and numerical results are presented for different Limit Order Book data: Apple, Amazon, Google, Microsoft, Intel on 2012/06/21 and Cisco, Facebook, Intel, Liberty Global, Liberty Interactive, Microsoft, Vodafone from 2014/11/03 to 2014/11/07. \\\\
		\textbf{Keywords:} limit order book; diffusion limit; market microstructure.
	\end{abstract}
	

	
	\section{Introduction}
	
	One of the main approaches of modeling Limit Order Books is the zero intelligence approach (see \cite{survey}), which assumes all quantities of interest in the Limit Order Book are governed by stochastic processes. Of the zero-intelligence models developed so far, the approach of \cite{ContLarrad} is an attractive starting point for modeling limit order flow in continuous time due to the tractability of the model and it's reduced dimensionality. They calculate various quantities of interest such as the probability of a price increase or the diffusion limit of the price process.\\
	
	Having found evidence in empirical observations, the authors of \cite{nelson2015b} extended the framework of \cite{ContLarrad}. They incorporated an arbitrary distribution for the inter-arrival times of the book events as well as a dependency of both, the type of a book event and its corresponding inter-arrival times, on the type of the previous book event. Therefore they used a Markov renewal process to model the dynamics of the bid and ask queues, which are assumed to be independent of each other. After a price change they are reinitialized. The model remains analytically tractable. As in \cite{ContLarrad} the bid/ask spread remains equal to one tick and all orders have the same size.
	
	We briefly recap the model used in \cite{nelson2015b} highlighting the notations and definitions we will use in this paper. As the model from \cite{ContLarrad}, price changes are assumed to happen at each time $T_n$ at which the ask or the bid queue is depleted. The sojourn times are notated as $\tau_n:=T_n-T_{n-1}$. The changes in the queue sizes are modeled by a Markov chain $V^a$ for the ask side and $V^b$ for the bid side. Their state space is $\{-1,1\}$. When a limit order appears at time t, the queue size increases by one unit and $V^a_t=1$, when a market order or cancellation appears, it decreases by one and $V^a_t=-1$. The notations for the bid side are defined accordingly.
	
	The paper defines a balanced and a unbalanced case. The classification is dependent on the transition probabilities of the Markov chain modeling the queue sizes:
	\[P^a(i,j):=P[V_{k+1}^a=j|V_k^a=i],\text{ }i,j\in \{-1,1\}.\]
	$P_t^b$ is defined accordingly.
	
	The balanced case is defined in the following way: $P^a(1,1)=P^a(-1,-1)$ and $P^b(1,1)=P^b(-1,-1)$. The unbalanced case is on hand if $P^a(1,1)<P^a(-1,-1)$ or $P^b(1,1)<P^b(-1,-1)$.
	
	Built on this model, \cite{nelson2015b} proposes the following jump model for the stock price $s_t$ based on a counting process $N(t)$ and a Markov Chain $X_t$:
	\begin{align*}
	s_t=\sum_{k=1}^{N(t)}X_k,
	\end{align*}
	where $N(t)$ counts how often the price changes and $X_t$ keeps track of which direction the price changed at each time of $T_n$, meaning $X_t \in \{-\delta,\delta\}$. A price change is assumed to happen at every time the bid or the ask queue of the Limit Order Book is depleted.
	
	We note that results of \cite{nelson2016b} were fist announced at IPAM FMWSI, UCLA, March 23-27, 2015 (see \cite{nelson2015a}). Also available at SRRN (see \cite{nelson2015b}) and arXiv (see \cite{nelson2016a}). The paper was submitted to SIAM Journal of Financial Mathematics in April 2015 and an extended version was resubmitted in March 2016.
	
	\subsection{Motivation for generalizing the model}
	
	As in \cite{ContLarrad}, \cite{nelson2015b} assume that all price changes occurring in the price process are of magnitude $\delta$, a single tick. Table  \ref{price_changesMID} demonstrates the average price change for each stock in our data set. For Apple midprice data we observe 53,654 price changes, from which only 9007 are of magnitude $\delta$. Meaning, 83.21\% (44,647) of the price changes were different from one tick. 
	
	\setlength{\tabcolsep}{2mm}
	\begin{table}[h]\centering
		\begin{tabular}{lccc}
			& Apple & Amazon & Google \\ \hline
			Avg. up movements & \jumpminus\jump1.7 & \jumpminus\jump1.3 & \jumpminus\jump3.1  \\ 
			Avg. down movements & \jump-1.7 & \jump-1.3 &\jump -3.0  \\ 
			Min price change & -18.5 & -11.5 & -30.5  \\ 
			Max price change& \jumpminus15.0 & \jumpminus16.5 & \jumpminus30.5 \\
			\hline
		\end{tabular}
		\caption{Mid-Price Changes in ticks} \label{price_changesMID}
	\end{table}
	
	Possible extensions of the proposed model were discussed in \cite{nelson2016b} for the case when the size of price changes is not fixed, including the diffusion limit of the price dynamics in this case. We will illustrate their model with the corresponding proofs and another model extension generalizing the model to allowing for more than two price changes.
	
	For both model extensions we show results for the diffusion coefficients and verify them using the same approach as in \cite{ContLarrad}.
	
	\subsection{Data}
	
	To test empirical validity of our more general model we use the following freely available data:
	
	\begin{itemize}
		\item Level 1 LOB data provided by \cite{lobster}: Apple, Amazon, Google, Microsoft and Intel on 2012/06/21
		\item LOB data provided in \cite{Book}: Cisco, Facebook, Intel, Liberty Global, Liberty Interactive, Microsoft, Vodafone from 2014/11/03 to 2014/11/07
	\end{itemize}
	
	\subsection{Structure of this paper}
	
	The rest of the paper is organized as follows. Section 2 reviews some of the assumptions from \cite{ContLarrad} and \cite{nelson2015b} with respect to our new data sets. The following sections 3 and 4 present two model extensions which generalize the model. Section 3 considers the general semi-Markov model for the Limit Order Books proposed by \cite{nelson2016b}. It incorporates two states modeling price changes that are not fixed to one tick. The section includes diffusion limits (sec. 3.1), implementations (sec. 3.2) and numerical results for \cite{lobster} data (sec. 3.3). Section 4 deals with an even more general semi-Markov model for Limit Order Books that incorporates an arbitrary number of states for the price changes. It includes a justification (sec. 4.1), diffusion limits (sec. 4.2), implementations (sec. 4.3) and numerical results for \cite{lobster} data (sec. 4.4). Section 5 discusses some empirical findings regarding the spread. Section 6 concludes the paper and highlights future work.
	
	\section{Reviewing the assumptions with our new data sets}
	
	The objective of this section is to test the validity of the assumptions of \cite{nelson2015b} and \cite{ContLarrad} with respect to new data. 
	
	\subsection{Liquidity of our data}
	
	Table \ref{liquidity} demonstrates the liquidity of the new data. Note that in calculating the average number of orders occurring in 10 seconds, we limited the orders to only those occurring at the best bid or ask price, i.e. at level 1.
	
	\begin{table}[h]\centering
		\begin{tabular}{lcc}
			& Average no. of orders in 10s & Price Changes in 1 day\\
			\hline
			AAPL & \jump51 & 64,350\\
			AMZN & \jump25 & 27,557\\
			GOOG & \jump 21 & 24,084\\
			INTC & 173 & \jump3,217\\
			MSFT & 176 & \jump4,060\\
			\hline
		\end{tabular}
		\caption{Stock liquidity of AAPL, AMZN, GOOG, INTC, and MSFT on 2012/06/21 \label{liquidity}}
	\end{table}	
	
	While the average number of orders in 10 seconds is significantly less than that reported in \cite{ContLarrad}, we note that most of these equities undergo more price changes in one day. The high number of daily price changes implies that we can use asymptotic analysis in order to approximate long-run volatility using order flow by finding the diffusion limit of the price process.

	\subsection{Empirical distributions of initial queue sizes and calculated conditional probabilities}
	
	Like \cite{ContLarrad} and \cite{nelson2015b}, the generalized model uses empirical distribution functions $f(q_t^b,q_t^a)$ and $\tilde{f}(q_t^b,q_t^a)$ to initialize the bid and ask queues after either a price increase, or decrease. In order to generalize their models, we trim our data to only points where the spread is a single tick. The resulting empirical distribution $f(q_t^b,q_t^a)$ for INTC from \cite{lobster} is displayed in figure \ref{empirical} (left).
	
	One of the accomplishments of \cite{ContLarrad} is the formula for the conditional probability of a price increase conditional on the state of the order book,
	$$ p_1^{up}(n,p) = \dfrac{1}{\pi}\int_{0}^{\pi} \left( 2 - \cos(t) - \sqrt{(2-cos(t))^2 - 1}\right)^p \dfrac{\sin(nt)\cos\left(\frac{t}{2}\right)}{\sin\left(\frac{t}{2}\right)}dt. $$
	We can compare this to the  quantity calculated in \cite{nelson2015b} as
	$$ p_1^{up}(n,p) = \int_{0}^{\infty}f_{p,a}(t)(1-F_{n,b}(t))dt$$
	where 
	\begin{align*}
	f_{p,a}(t) & = \dfrac{1}{2\pi}\int_{\mathbb{R}} e^{itx}\varphi^a(x,p)dx\\ 
	F_{n,b}(t) & = \dfrac{1}{2}- \dfrac{1}{\pi}\int_{0}^{\infty} \dfrac{1}{x}Im\{e^{-itx}\varphi^b(x,n)\}dx\\
	\end{align*}
	and $\varphi^a(t), \varphi^b(t)$ are the characteristic functions of $\sigma_a$ and $\sigma_b$, respectively.
	
	Figure \ref{empirical} (right) shows, for INTC, the conditional probability of a price increase conditional on the size of bid and ask queues, as calculated using the formula from \cite{ContLarrad}.
	
	\begin{figure}[t]
		\subfigure{\includegraphics[width=0.49\textwidth]{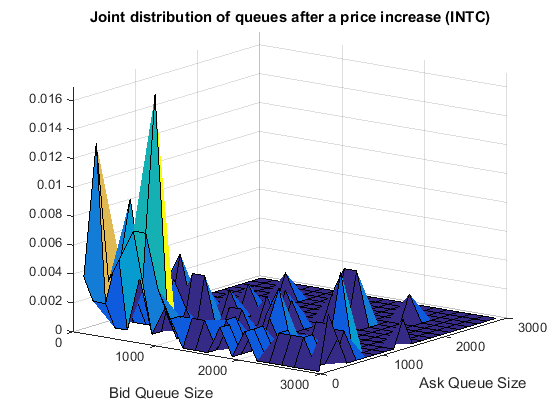}}
		\subfigure{\includegraphics[width=0.49\textwidth]{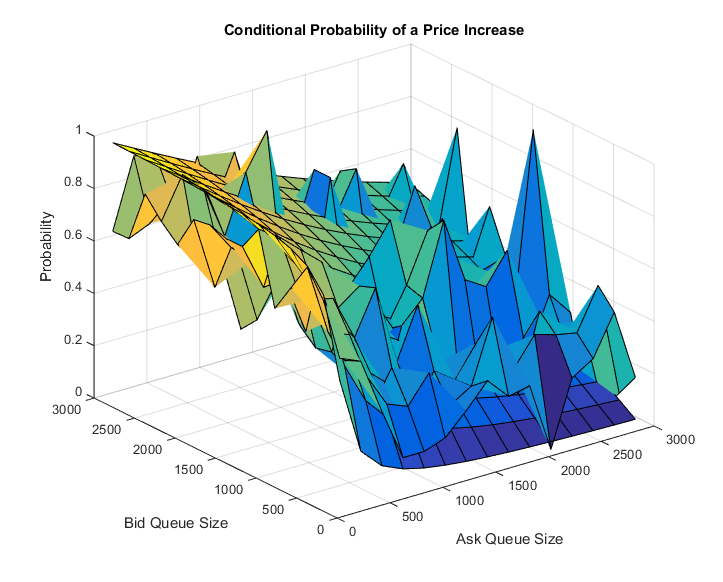}}
		\vspace*{8pt}
		\caption{Left: Empirical Joint Density after a price increase,  $f(q_t^b,q_t^a)$  Right: Conditional probability of a price increase conditional on size of bid and ask queues}\label{empirical}
	\end{figure}
	
	While the empirical data does not fit $p_1^{up}(n,p)$ as closely as in \cite{ContLarrad}, we can see that the calculated value from the model still matches the empirical frequencies to some extent.
	
	\subsection{Inter-arrival times of book events}
	
	While \cite{ContLarrad} assume the inter-arrival times between book events follow independent exponential distributions, \cite{nelson2015b} challenges this assumption. We have calculated the empirical distribution functions of the inter-arrival times between book events for our data and get the same result as \cite{nelson2015b}: The exponential distribution does not fit the data as well as alternative distributions do. Figure \ref{interarrival} illustrates the example of Amazon. We have similar figures for Apple, Google, Intel and Microsoft illustrating the same finding.
	
	\begin{figure}[t]
		\subfigure{\includegraphics[width=0.49\textwidth]{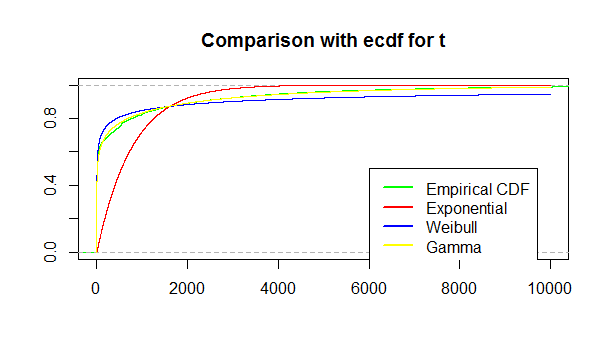}}
		\subfigure{\includegraphics[width=0.49\textwidth]{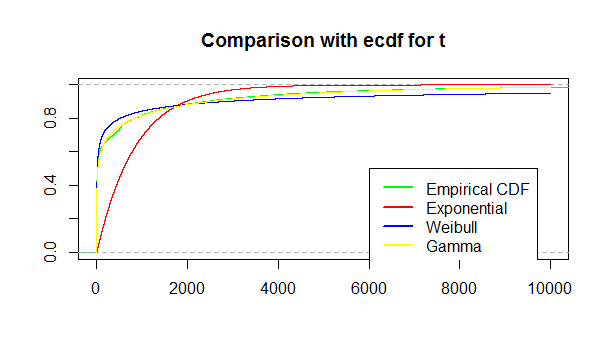}}
		\vspace*{8pt}
		\caption{Distribution of inter-arrival times Amazon ask and bid}\label{interarrival}
	\end{figure}
	
	\subsection{Asymptotic analysis}
	
	The asymptotic analysis presented in \cite{ContLarrad} is another strength of their paper. Using the relevant formula when the rate of incoming limit orders is assumed to equal the combined rate of incoming market orders and cancellations, we compute the diffusion coefficients for our new data. \cite{ContLarrad} demonstrates the linear relationship between $\sqrt{\frac{\lambda}{D(f)}}$ and the 10 minute standard deviation of various equities, where $\lambda$ is the intensity of incoming orders and $D(f)$ is the square of the average depth of the bid and ask queues after a price change. Figure \ref{linearRel} contains the same plot generated for the new data. The linear relationship between $\sqrt{\frac{\lambda}{D(f)}}$ and the 10 minute standard deviation is not as significant for the new data set. We will see later on a huge improvement for our considered model extensions.
	\begin{figure}[t]
		\vspace*{-7pt}
		\centerline{\includegraphics[width=0.55 \textwidth]{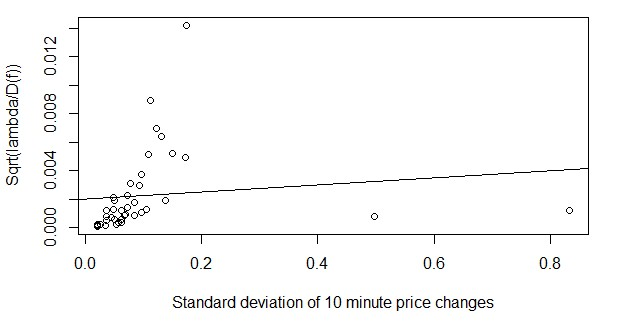}}
		\vspace*{8pt}
		\caption{$\sqrt{\lambda/D(f)}$ compared to 10 minute standard deviation}
		\label{linearRel}
	\end{figure}
	
	\section{General Semi-Markov Model for the Limit Order Book with two states \label{Two_States}}\label{twostates}
	
	We develop the model proposed in \cite{nelson2016b} (see sec. 6: Discussion), incorporating price changes which are not fixed to one tick. The authors already introduced the following model for the price process:
	
	\begin{align*}
	s_t=\sum_{k=1}^{N(t)}a(X_k),
	\end{align*}
	
	where $N(t)$ is the counting process for the price changes, $X_k$ is a two state Markov chain with state space $\mathcal{S}=\{1,2\}$ and $a(x)$ is an uniformly bounded function. To consider one tick spreads we set $a(1)=\delta$ and $a(2)=-\delta$.
	
	\subsection{Diffusion Limits}\label{Diffusion_Limits}
	
	For the calculation of the diffusion limits we will use the following two assumptions from \cite{nelson2015b}.
	
	\begin{ass}\label{A4}
		We assume the following inequalities to be true:
		\[\sum_{n=1}^{\infty}\alpha^b(n)\alpha^a(p)f(n,p)<\infty,\]
		\[\sum_{n=1}^{\infty}\alpha^b(n)\alpha^a(p)\tilde{f}(n,p)<\infty,\]
		with \(\alpha^a(n):=\frac{1}{p_a\sqrt{\pi}}(n+\frac{2p_a-1}{p_a-1}v_0^a(1))\sqrt{p_a(1-p_a)}\sqrt{p_ah_1^a+(1-p_a)h_2^a},\) $h_1^a$ is defined on page 7 in \cite{nelson2015b}, $p_a=P^a(1,1)=P^a(-1,-1)$, $\alpha^b(n)$ is defined accordingly.
	\end{ass}
	
	\begin{ass}\label{A5}
		In this section we assume
		$m(1)<\infty$ and $ m(2)<\infty$, where $m(i)=E[\tau_k|X_{k-1}=i], i\in\{1,2\}$.
	\end{ass}
	
	\begin{thm}\label{Theorem1}
		Given that assumption \ref{A4} is satisfied for the balanced case and assumption \ref{A5} for the unbalanced case, we can proof the following weak convergences in the Skorokhod topology (see \cite{Skorokhold}):
		\begin{center}
			$\left(\dfrac{s_{tnlog(n)}-N_{tnlog(n)}a^*}{\sqrt{n}},t\ge 0\right)\overset{n\rightarrow \infty}{\Rightarrow} \dfrac{\sigma^*}{\sqrt{\tau^*}}W_t$, for the balanced case and \\
			$\left(\dfrac{s_{tn}-N_{tn}a^*}{\sqrt{n}},t\ge 0\right)\overset{n\rightarrow \infty}{\Rightarrow} \dfrac{\sigma^*}{\sqrt{m_{\tau}}}W_t$, for the unbalanced case, \\
		\end{center}
		where \(W_t\) is a standard Brownian motion, $a_i:=a(i)$, \(a^*=\pistar_1a_1+\pistar_2a_2\)and \((\pistar_1,\pistar_2)\) is the stationary distribution of the Markov chain $a(X)$. $\tau^*$, $m_{\tau}$ and $(\sigma^*)^2$ are given by:
		\begin{align*}
		\tau^*&=\lim_{t\rightarrow+\infty}\dfrac{t}{N_tlog(N_t)} \text{ (see \cite{nelson2015b}, p.19)}\\
		m_{\tau}&=\pistar_1m(1)+\pistar_2m(2)\\
		(\sigma^*)^2&=\pistar_1a_1^2+\pistar_2a_2^2+(\pistar_1a_1+\pistar_2a_2)[-2a_1\pistar_1-2a_2\pistar_2+(\pistar_1a_1+\pistar_2a_2)(\pistar_1+\pistar_2)]\\
		&+\dfrac{(\pistar_1(1-\p)+\pistar_2(1-\p'))(a_1-a_2)^2}{(\p+\p'-2)^2}\\
		&+2(a_2-a_1)\cdot\Bigg[\dfrac{\pistar_2a_2(1-\p')-\pistar_1a_1(1-\p)}{\p+\p'-2}\\
		&+\dfrac{(\pistar_1a_1+\pistar_2a_2)(\pistar_1-\p\pistar_1-\pistar_2+\p'\pistar_2)}{\p+\p'-2}\Bigg]
		\end{align*}
		
		$p_{cont}$ is the probability of two subsequent increases of the stock price, $p_{cont}'$ the probability of two subsequent decreases of the stock price.
		
	\end{thm}
	
	\begin{proof}
		We get the following law of large numbers:
		\begin{align*}
		\dfrac{s_{tnlog(n)}}{n} \overset{n\rightarrow \infty}{\Rightarrow} \dfrac{a^* t}{\tau^*}.
		\end{align*}
		The proof follows that of Proposition 8 in \cite{nelson2015b}, except of the calculation for $\sigma$.  For $t\in R_+$ we first of all consider the following processes:
		\[R_n:=\sum_{k=1}^n(a(X_k)-a^*),\]
		\[U_n(t):=n^{-1/2}[(1-\lambda_{n,t})R_{\lfloor{nt}\rfloor}+\lambda_{n,t}R_{\lfloor{nt}\rfloor+1}],\]
		where \(\lambda_{n,t}:= nt-\lfloor{nt}\rfloor\). We can show the following weak convergence in the Skorokhod topology (see \cite{Skorokhold}) similar to the approach of \cite{swishtheory}:
		\[(U_n(t),t\ge0)\overset{n\rightarrow \infty}{\Rightarrow}\sigma^* W,\]
		where W is a standard Brownian motion, and \(\sigma\) is given by:
		\[(\sigma^*)^2:=\sum_{i\in\{1,2\}}\pi^*_iv^*(i)\]
		\[ v^*(1):=(a_1-a^*)^2+p(1)(g_2-g_1)^2-2(a_1-a^*)p(1)(g_2-g_1),\]
		\[ v^*(2):=(a_2-a^*)^2+p(2)(g_1-g_2)^2-2(a_2-a^*)p(2)(g_1-g_2),\]
		\[
		\begin{pmatrix}
		g_1\\g_2
		\end{pmatrix}
		=(P+\Pi^*-I)^{-1}
		\begin{pmatrix}
		a_1-a^* \\
		a_2-a^*
		\end{pmatrix},\]
		\[p(1):=1-p_{cont}, p(2)=1-p_{cont}',\]
		\(\Pi^*\) is the matrix of the stationary distribution consisting of rows equal to \((\pi^*_1 \text{     	      } \pi^*_2).\)\\
		
		We get the following calculation for g:
		\begin{align*}
		g:=&(P+\Pi^*-I)^{-1}(a-a^*)\\
		=&
		\left[
		\begin{pmatrix}
		\p & 1-\p\\
		1-\pi' & \p'
		\end{pmatrix}
		+
		\begin{pmatrix}
		\pistar_1 & \pistar_2\\
		\pistar_1 & \pistar_2
		\end{pmatrix}
		-
		\begin{pmatrix}
		1&0\\
		0&1
		\end{pmatrix}
		\right]^{-1}
		\begin{pmatrix}
		a_1-a^*\\
		a_2-a^*
		\end{pmatrix}\\
		=&\dfrac{1}{(\p+\pistar_1-1)(\p'+\pistar_2-1)-(1-\p+\pistar_2)(1-\p'+\pistar_1)}\\
		\cdot&\begin{pmatrix}
		\p+\pistar_2-1 & \p-\pistar_2-1\\
		\p'-\pistar_1-1 & \p+\pistar_1-1
		\end{pmatrix}
		\begin{pmatrix}
		a_1-a^*\\
		a_2-a^*
		\end{pmatrix}\\
		=&\begin{pmatrix}
		\dfrac{a_1(p_{cont}'+\pi^*_2-1)+a_2(p_{cont}-\pi^*_2-1)}{(\pi^*_1+\pi^*_2)(p_{cont}+p_{cont}'-2)}-\dfrac{a^*}{\pi^*_1+\pi^*_2}\\[1.5em]
		\dfrac{a_1(p_{cont}'-\pi^*_1-1)+a_2(p_{cont}+\pi^*_1-1)}{(\pi^*_1+\pi^*_2)(p_{cont}+p_{cont}'-2)}-\dfrac{a^*}{\pi^*_1+\pi^*_2}
		\end{pmatrix}
		\end{align*}\\
		
		In order to get a nice form for $(\sigma^*)^2$ we firstly calculate the summand for $i=1$:
		
		\begin{align*}
		\pistar_1v(1)&=\pistar_1 \left[(a_1-a^*)^2+(1-\p)\left(\dfrac{a_1(\p'-\pistar_1-1)+a_2(\p+\pistar_1-1)}{(\p+\p'-2)(\pistar_1+\pistar_2)}\right.\right.\\
		&\left.\left.-\dfrac{a^*}{\pistar_1+\pistar_2}
		-\dfrac{a_1(\p'+\pistar_2-1)+a_2(\p-\pistar_2-1)}{(\p+\p'-2)(\pistar_1+\pistar_2)}+\dfrac{a^*}{\pistar_1+\pistar_2}\right)^2\right.\\
		&\left.-2(a_1-a^*)(1-\p)\left(\dfrac{a_1(\p'-\pistar_1-1)+a_2(\p+\pistar_1-1)}{(\p+\p'-2)(\pistar_1+\pistar_2}-\dfrac{a^*}{\pistar_1+\pistar_2}\right.\right.\\
		&\left.\left.-\dfrac{a_1(\p'+\pistar_2-1)+a_2(\p-\pistar_2-1)}{(\p+\p'-2)(\pistar_1+\pistar_2)}+\dfrac{a^*}{\pistar_1+\pistar_2}\right)\right]\\
		&=\pistar_1\left[(a_1-a^*)^2+(1-\p)\left(\dfrac{a_2-a_1}{\p+\p'-2}\right)^2\right.\\
		&\left.-2(a_1-a^*)(1-\p)\left(\dfrac{a_2-a_1}{\p+\p'-2}\right)\right]\\
		&=\pistar_1\left[(a_1-(\pistar_1a_1+\pistar_2a_2))^2+(1-\p)\dfrac{(a_2-a_1)^2}{(\p+\p'-2)^2}\right.\\
		&\left.-2(a_1-(\pistar_1a_1+\pistar_2a_2))(1-\p)\dfrac{a_2-a_1}{\p+\p'-2}\right]
		\end{align*}
		Similarly we calculate the summand for $i=2$:
		\begin{align*}
		\pistar_2v(2)&=\pistar_2\left[(a_2-a^*)^2+(1-\p')
		\left(\dfrac{a_1(\p'+\pistar_2-1)+a_2(\p-\pistar_2-1)}{(\p+\p'-2)(\pistar_1+\pistar_2)}\right.\right.\\
		&\left.\left.-\dfrac{a^*}{\pistar_1+\pistar_2}-\dfrac{a_1(\p'-\pistar_1-1)+a_2(\p+\pistar_1-1)}{(\p+\p'-2)(\pistar_1+\pistar_2)}+\dfrac{a^*}{\pistar_1+\pistar_2}\right)^2\right.\\
		&\left.-2(a_2-a^*)(1-\p')\left(\dfrac{a_1(\p'+\pistar_2-1)+a_2(\p-\pistar_2-1)}{(\p+\p'-2)(\pistar_1+\pistar_2)}-\dfrac{a^*}{\pistar_1+\pistar_2}\right.\right.\\
		&\left.\left.-\dfrac{a_1(\p'-\pistar_1-1)+a_2(\p+\pistar_1-1)}{(\p+\p'-2)(\pistar_1+\pistar_2)}+\dfrac{a^*}{\pistar_1+\pistar_2}\right)\right]\\
		&=\pistar_2\left[(a_2-a^*)^2+(1-\p')\left(\dfrac{a_1-a_2}{\p+\p'-2}\right)^2\right.\\
		&\left.-2(a_2-a^*)(1-\p')\left(\dfrac{a_1-a_2}{\p+\p'-2}\right)\right]\\
		&=\pistar_2\left[(a_2-(\pistar_1a_1+\pistar_2a_2))^2+(1-\p')\dfrac{(a_1-a_2)^2}{(\p+\p'-2)^2}\right.\\
		&\left.-2(a_2-(\pistar_1a_1+\pistar_2a_2))(1-\p')\dfrac{a_1-a_2}{\p+\p'-2}\right]
		\end{align*}
		From this it follows:
		\begin{align*}
		\sigma^2&=\pistar_1v(a_1)+\pistar_2v(a_2)\\
		&=\pistar_1\left(a_1^2-2a_1(\pistar_1a_1+\pistar_2a_2)+(\pistar_1a_1+\pistar_2a_2)^2\right)\\
		&+\pistar_2\left(a_2^2-2a_2(\pistar_1a_1+\pistar_2a_2)+(\pistar_1a_1+\pistar_2a_2)^2\right)\\
		&+\dfrac{(\pistar_1(1-\p)+\pistar_2(1-\p'))(a_1-a_2)^2}{(\p+\p'-2)}\\
		&+2(a_2-a_1)\left[\dfrac{\pistar_2a_2(1-\p')-\pistar_1a_1(1-\p)}{\p+\p'-2}\right.\\
		&\left.+\dfrac{(\pistar_1a_1+\pistar_2a_2)(\pistar_1-\p\pistar_1-\pistar_2+\p'\pistar_2)}{\p+\p'-2}\right]\\
		&=\pistar_1a_1^2+\pistar_2a_2^2+(\pistar_1a_1+\pistar_2a_2)[-2a_1\pistar_1-2a_2\pistar_2+(\pistar_1a_1+\pistar_2a_2)(\pistar_1+\pistar_2)]\\
		&+\dfrac{(\pistar_1(1-\p)+\pistar_2(1-\p'))(a_1-a_2)^2}{(\p+\p'-2)^2}\\
		&+2(a_2-a_1)\left[\dfrac{\pistar_2a_2(1-\p')-\pistar_1a_1(1-\p)}{\p+\p'-2}\right.\\
		&\left.+\dfrac{(\pistar_1a_1+\pistar_2a_2)(\pistar_1-\p\pistar_1-\pistar_2+\p'\pistar_2)}{\p+\p'-2}\right]
		\end{align*}\\\\
		The continuing proof directly follows the proof of Proposition 8 and Proposition 10 in \cite{nelson2015b}.
	\end{proof}
	
	\begin{rmk}
		When inserting $a_1=\delta, a_2=-\delta$, we get $a^*=s^*, \pistar_1=\pistar$ and $\pistar_2=1-\pistar$ and therewith 
		
		\begin{align*}
		\sigma^2=4\delta^2\left(\dfrac{1-\p'+\pistar(\p'-\p)}{(\p+\p'-2)^2}-\pistar(1-\pistar)\right).
		\end{align*}\label{calc_sig_alt}
		This is the same result as in Proposition 8 in \cite{nelson2015b}, which shows that our model is a generalization of their model.
	\end{rmk}
	
	\begin{proof}
		When inserting $a_1=\delta, a_2=-\delta$ in $\sigma^2$ we get $a^*=s^*, \pistar_1=\pistar$ and $\pistar_2=1-\pistar$ and therewith
		\begin{align*}
		\sigma^2&=\pistar\delta^2+(1-\pistar)(-\delta)^2+(\pistar\delta+(1-\pistar)(-\delta))(-2\pistar\delta-2(1-\pistar)(-\delta)\\
		&+(\pistar\delta+(1-\pistar)(-\delta))(\pistar+(1-\pistar)))\\
		&+\dfrac{(\pistar(1-\p)+(1-\pistar)(1-\p'))(\delta-(-\delta))^2}{(\p+\p'-2)^2}\\
		&+2(-\delta-\delta)\left[\dfrac{(1-\pistar)(-\delta)(1-\p')-\pistar\delta(1-\p)}{\p+\p'-2}\right.\\
		&\left.+\dfrac{(\pistar\delta+(1-\pistar)(-\delta))(\pistar-\p\pistar-(1-\pistar)+(1-\pistar)\p')}{\p+\p'-2}\right]\\
		&=\delta^2+(2\pistar\delta-\delta)(-2\pistar\delta+\delta)+4\delta^2\dfrac{1-\p'+\pistar\p'-\pistar\p}{(\p+\p'-2)^2}\\
		&-4\delta^2\dfrac{4(\pistar)^2-4\pistar-2(\pistar)^2\p-2(\pistar)^2\p'+2\pistar\p+2\pistar\p'}{\p+\p'-2}\\
		&=4\pistar\delta^2-4(\pistar)^2\delta^2+4\delta^2\left(\dfrac{1-\p'+\pistar(\p'-\p)}{(\p+\p'-2)^2}\right.\\
		&\left.-\dfrac{2\pistar(2\pistar-2-\pistar\p-\pistar\p'+\p+\p')}{\p+\p'-2}\right)\\
		&=4\delta^2\left(\pistar-(\pistar)^2+\dfrac{1-\p'+\pistar(\p'-\p)}{(\p+\p'-2)^2}\right.\\
		&\left.-\dfrac{2\pistar(1-\pistar)(\p+\p'-2)}{\p+\p'-2}\right)\\
		&=4\delta^2\left(\dfrac{1-\p'+\pistar(\p'-\p)}{(\p+\p'-2)^2}+\pistar-(\pistar)^2-2\pistar+2(\pistar)^2\right)\\
		&=4\delta^2\left(\dfrac{1-\p'+\pistar(\p'-\p)}{(\p+\p'-2)^2}-\pistar(1-\pistar)\right).
		\end{align*}
	\end{proof}
	
	\subsection{Implementation}\label{implementation2states}
	This subsection explains how we realized the implementation of the diffusion limits given in subsection \ref{Diffusion_Limits}.
	
	It is clear that assumptions \ref{A4} and \ref{A5} are fulfilled, as we are considering finite data sets. Therefore the calculations are straightforward.
	
	The question arises how to define the function $a(X_k)$. It is enough to define the values of $a_1$ and $a_2$, in lieu of the whole function $a$. We set $a_1$ to the average of upward movements and $a_2$ to the one of downwards movements respectively.
	
	The matrix $P$, which is defined as 
	\[P=\begin{pmatrix}
	\p & 1-\p\\
	1-\p' & \p'
	\end{pmatrix},\]
	contains the transition probabilities of the Markov chain $X_k$. We assign each positive stock price change to the state 1 and each negative one to 2. We count their absolute frequencies and calculate the relative frequencies to get $p_{cont}$ and $p_{cont}'$.
	
	The stationary distribution $\pistar$ of the transition matrix P satisfies $\pistar=\pistar P$ and the exclusively positive entries have to sum up to 1 since it is a probability distribution. This is equivalent to solving the problem
	\[\begin{pmatrix}
	\p-1 & 1-\p'\\
	1 & 1
	\end{pmatrix}
	\begin{pmatrix}
	\pistar_1\\
	\pistar_2
	\end{pmatrix}
	=
	\begin{pmatrix}
	0\\
	1
	\end{pmatrix}
	.\]
	We need to remark at that point why the stationary distribution exists. As stated in several books about the theory of Markov chains, as e.g. in \cite{markovchain}, a unique stationary distribution exists for a Markov chain with a finite space set as long as the Markov chain is irreducile. The state space is finite for every set of our data and looking at the transition matrices, we can see that there are no closed sets and that the Markov chain is irreducible.
	
	We calculate \(\tau^*\) by using the following result from Lemma 7 in \cite{nelson2015b}: 
	\[\frac{1}{nlog(n)}\sum_{k=1}^n \tau_k\rightarrow\tau^*,\] where n is the total number of price changes. 
	
	For the computation of \(m_{\tau}\), $\sigma^*$ and the diffusion coefficients we use the formulas given in subsection \ref{Diffusion_Limits}. 
	
	\subsection{Numerical Results}\label{Num_Res}
	The following tables depict the results gained from our computations for the \cite{lobster} data for Apple, Amazon, Google, Intel and Microsoft on 06/21/2012. We do not consider the first and last 15 minutes after opening and before closing, because we do not want to include opening and closing auctions. 
	
	\setlength{\tabcolsep}{1mm}
	\begin{table}[h]\centering
		\begin{tabular}{@{}lccccccccc@{}}
			& $p_{cont}$ & $p_{cont}'$ & $a_1$ & $a_2$ & $\sigma^2$ & $\tau^*$ & $m_\tau$ & $\dfrac{\sigma}{\sqrt{\tau^*}}$ & $\dfrac{\sigma}{\sqrt{m_\tau}}$\\ \hline
			Apple & 0.4932 & 0.4956 & 0.0170 & -0.0172 & 0.0003\jump & 0.0370 & \jump0.4026 & 0.0881 & 0.0267 \\ 
			Amazon & 0.4576 & 0.4635 & 0.0133 & -0.0134 & 0.0002\jump & 0.0892 & \jump0.9001 & 0.0412 & 0.0130 \\ 
			Google & 0.4461 & 0.4769 & 0.0308 & -0.0302 & 0.0008\jump & 0.1145 & \jump1.1291 & 0.0834 & 0.0266 \\ 
			Intel & 0.5588 & 0.6106 & 0.0050 & -0.0050 & 0.00004 & 1.3151 & 10.0897 & 0.0052 & 0.0019  \\
			Microsoft & 0.5827 & 0.6269 & 0.0050 & -0.0050 & 0.00004 & 0.8944 & \jump7.1657 & 0.0065 & 0.0023 \\ \hline
		\end{tabular}
		\caption{Numerical Results for two states}
	\end{table}
	
	In subsections 4.2 and 4.3 in \cite{ContLarrad} the authors state that their diffusion coefficients are linearly related to the standard deviation of the ten minute price changes. Therefore we calculated these standard deviations for all our data. We get an adjusted R\textsuperscript{2} of 0.9788 for the linear regression of the diffusion coefficients on the standard deviation for the balanced case and one of 0.9821 for the unbalanced case. From this it follows, that the standard deviation of the ten minutes price changes is linearly dependent on the diffusion coefficients. In both cases the regression coefficients are highly significant.
	
	\begin{figure} [t]
		\subfigure{\includegraphics[width=0.49\textwidth]{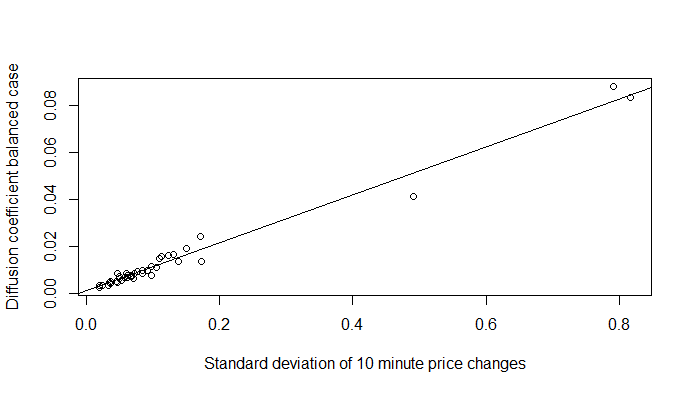}} 
		\subfigure{\includegraphics[width=0.49\textwidth]{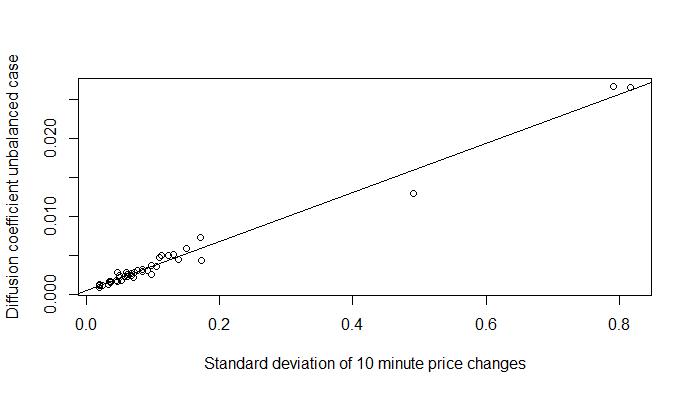}}
		\vspace*{8pt}
		\caption{Linear relationship of coefficients and standard deviation of 10 minutes mid-price changes for two states balanced case (left) and unbalanced case (right)} 
		\label{ticks}
	\end{figure} 
	
	As a comparison and justification for our model, we also calculated the regression for the diffusion coefficients resulting from the basic model of \cite{nelson2015b} allowing only for price changes of one tick. We used exactly the same data as for the regression above. The regression is plotted in figure \ref{one_tick}. The adjusted R\textsuperscript{2} for the balanced case is 0.3916, the one for the unbalanced case is 0.3813. It can clearly be seen that the linear relationship is better captured by our extended model.
	
	\begin{figure} [t]
		\subfigure{\includegraphics[width=0.49\textwidth]{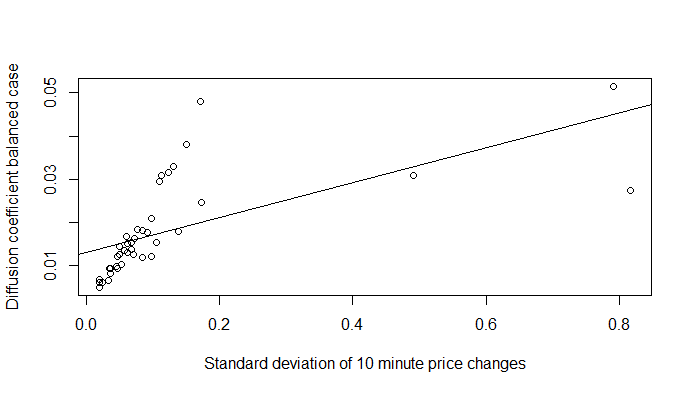}} 
		\subfigure{\includegraphics[width=0.49\textwidth]{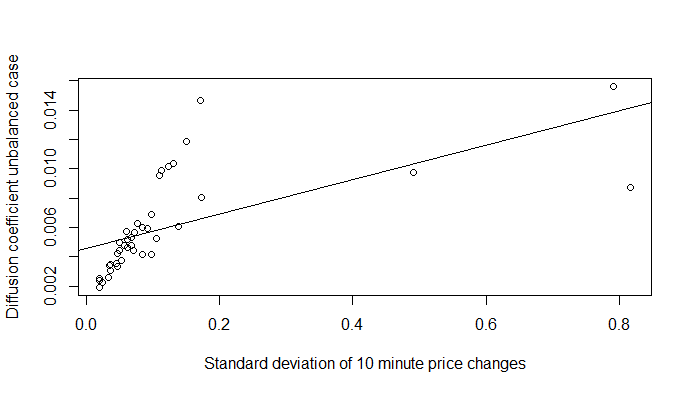}}
		\vspace*{8pt}
		\caption{Linear relationship of coefficients and standard deviation of 10 minutes mid-price changes for one tick jump sizes balanced case (left) and unbalanced case (right)} 
		\label{one_tick}
	\end{figure} 
	
	\section{General Semi-Markov Model for the Limit Order Book with arbitrary number of states}\label{severalstates}
	
	\subsection{Justification}
	
	Our next goal is to further generalize the modeling of the stock price. In the last section we assumed that the jump sizes of the stock prices can only take two values $a(1)$ and $a(2)$. Some of the available data give evidence that the price changes can take more than two values. This can be seen in figure \ref{AppleGoogle}. The exact numbers of different price changes for Apple, Amazon and Google are stated in table \ref{summary}, which justify the modeling of more than two states. 
	\setlength{\tabcolsep}{2mm}
	\begin{table}[h]\centering
		\begin{tabular}{lccc}
			& Apple & Amazon & Google \\ \hline
			States & 60 & 46 & 87\\ \hline
		\end{tabular}
		\caption{Number of different price changes of mid prices \label{summary}}
	\end{table}
	
	\begin{figure}[t]
		\centering
		\subfigure{\includegraphics[width=0.32\textwidth]{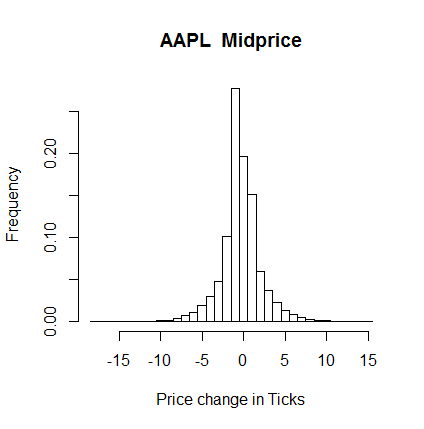}}
		\subfigure{\includegraphics[width=0.32\textwidth]{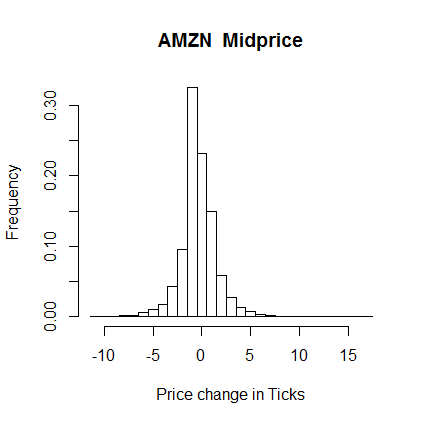}}
		\subfigure{\includegraphics[width=0.32\textwidth]{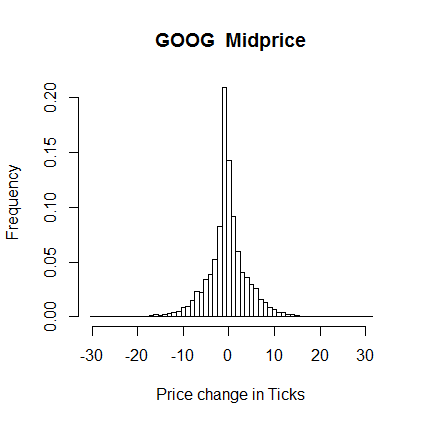}}
		\vspace*{8pt}
		\caption{Jump sizes Apple, Amazon and Google Midprices}
		\label{AppleGoogle}
	\end{figure}
	
	We will consider the following model:
	
	\begin{align*}
	s_t=\sum_{k=1}^{N(t)} a(X_k),
	\end{align*}
	where $X_k$ is a Markov chain with n states, meaning that the state space is extended to $\mathcal{S}=\{1,...,n\}$.
	
	\subsection{Diffusion Limits}
	
	For the extended model, balanced and unbalanced case are defined as in subsection \ref{Diffusion_Limits}. Assumption \ref{A4} does not change as well.
	
	\begin{ass}\label{A7}
		We assume $m(i)<\infty \text{ for all } i=1,2,\dots n$, where $m(i)$ is defined like in subsection \ref{Diffusion_Limits}.
	\end{ass}
	
	We get the following results for the diffusion limits in the new model:
	
	\begin{thm}
		
		Given that assumption \ref{A4} is satisfied for the balanced case and assumption \ref{A7} is satisfied for the unbalanced case, we can proof the following weak convergences in the Skorokhod topology (see \cite{Skorokhold}):
		
		\begin{center}
			$\left(\dfrac{s_{tnlog(n)}-N_{tnlog(n)}a^*}{\sqrt{n}},t\ge 0\right)\overset{n\rightarrow \infty}{\Rightarrow} \dfrac{\sigma^*}{\sqrt{\tau^*}}W_t$, for the balanced case and \\
			$\left(\dfrac{s_{tn}-N_{tn}a^*}{\sqrt{n}},t\ge 0\right)\overset{n\rightarrow \infty}{\Rightarrow} \dfrac{\sigma^*}{\sqrt{m_{\tau}}}W_t$, for the unbalanced case, \\
		\end{center}
		where \(W_t\) is a standard Brownian motion, \(a^*=\sum_{i \in \mathcal{S}}\pistar_i a(i)\) and \(m_\tau=\sum_{i \in \mathcal{S}}\pistar_i m(i)\). $\tau^*$ and $(\sigma^*)^2$ are given by:
		\[\tau^*=\lim_{t\rightarrow+\infty}\dfrac{t}{N_tlog(N_t)} \text{ (see \cite{nelson2015b}, p.19)} \]
		\[(\sigma^*)^2=\sum_{i \in \mathcal{S}} \pi_iv(i)\]
		\[v(i)= b(i)^2+\sum_{j\in\mathcal{S}}(g(j)-g(i))^2P(i,j)-2b(i)\sum_{j\in\mathcal{S}}(g(j)-g(i))P(i,j),\]
		where
		\begin{align*}
		b&=(b(1),b(2),...,b(n))',\\
		b(i):&=a(X_i)-a^*:=a(i)-a^* \text{ and} \\
		g:&=(P+\Pistar-I)^{-1}b.
		\end{align*}
		$P$ is a transition probability matrix, where $P(i,j)=P(X_{k+1}=j|X_k=i)$. $\Pistar$ denotes the stationary distribution of $P$ and $g(j)$ is the j\textsuperscript{th} entry of $g$.
		
	\end{thm}
	
	\begin{proof}
		We get the following law of large numbers:
		\begin{align*}
		\dfrac{s_{tnlog(n)}}{n} \overset{n\rightarrow \infty}{\Rightarrow} \dfrac{a^* t}{\tau^*}.
		\end{align*}
		
		The proof follows the one in \cite{nelson2015b}, in which the calculation for $\sigma$ has to be adapted. Therefore we consider the more general result given on page 28 in \cite{swishtheory}. For simplification we mainly use our notations instead of the ones they used.
		
		Denote for \(t\in R_+\) (as in \cite{nelson2015b}):
		\begin{equation*}
		R_n:=\sum_{k=1}^n(a(X_k)-a^*),
		\end{equation*}
		\begin{equation*}
		U_n(t):=n^{-1/2}[(1-\lambda_{n,t})R_{\lfloor{nt}\rfloor}+\lambda_{n,t}R_{\lfloor{nt}\rfloor+1}],
		\end{equation*}
		where \(\lambda_{n,t}:= nt-\lfloor{nt}\rfloor\). In \cite{swishtheory} (page 28) it is shown for a more general case that we have the following weak convergence in the Skorokhod topology:
		\[(U_n(t),t\ge0)\overset{n\rightarrow \infty}{\Rightarrow}\sigma W,\]
		where W is a standard Brownian motion, \(\sigma^2:=\sum_i\pi(i)v(i),\) and for \(i \in \mathcal{S}:\)
		
		\begin{align*}
		v(i)=&\sum_{j\in\mathcal{S}}\int_0^{\infty}(f(i,j,u)-\alpha_f(u))^2Q(i,j,du)+\sum_{j\in\mathcal{S}}(g(j)-g(i))^2P(i,j)\\
		-&2\sum_{j\in\mathcal{S}}(g(j)-g(i))P(i,j)\int_0^{\infty}(f(i,j,u)-\alpha_f(u))H(i,j,du).
		\label{compli}
		\end{align*}
		
		Q is the kernel of the Markov renewal process
		$Q(X_n,j,t)=\mathbb{P}[X_{n+1}=j,\\\tau_{n+1}\leq t|X_n]$
		and can also be written as $Q(i,j,t)=P(i,j)H(i,j,t)$.
		$P$ is the transition matrix of the Markov chain and \mbox{$H(i,j,t):=\mathbb{P}[\tau_{n+1}\leq t|X_n=i,X_{n+1}=j]$}.
		Further the result in \cite{swishtheory} contains time inhomogeneitey. As we consider time homogeneity, $\alpha_f$ is a constant equal to $\astar$ of section \ref{twostates}. Since the bounded function $a(\cdot)$ is only dependent on one variable, we simply have a look at the case where $f(i)=a(i)$.
		Inserting the known simplifications we get 
		\begin{align*}
		v(i)&=\sum_{j\in\mathcal{S}}\int_0^{\infty}(a(i)-\astar)^2Q(i,j,du)+\sum_{j\in\mathcal{S}}(g(j)-g(i))^2P(i,j)\\
		&-2\sum_{j\in\mathcal{S}}(g(j)-g(i))P(i,j)\int_0^{\infty}(a(i)-\astar)H(i,j,du)\\
		&=\sum_{j\in\mathcal{S}}\int_0^{\infty}b(i)^2P(i,j)H(i,j,du)+\sum_{j\in\mathcal{S}}(g(j)-g(i))^2P(i,j)\\
		&-2\sum_{j\in\mathcal{S}}(g(j)-g(i))P(i,j)\int_0^{\infty}b(i)H(i,j,du)\\
		&=b(i)^2\sum_{j\in\mathcal{S}}P(i,j)\int_0^{\infty}H(i,j,du)+\sum_{j\in\mathcal{S}}(g(j)-g(i))^2P(i,j)\\
		&-2\sum_{j\in\mathcal{S}}(g(j)-g(i))P(i,j)b(i)\int_0^{\infty}H(i,j,du)\\
		&=b(i)^2\sum_{j\in\mathcal{S}}P(i,j)+\sum_{j\in\mathcal{S}}(g(j)-g(i))^2P(i,j)-2\sum_{j\in\mathcal{S}}(g(j)-g(i))P(i,j)b(i)\\
		&=b(i)^2+\sum_{j\in\mathcal{S}}(g(j)-g(i))^2P(i,j)-2b(i)\sum_{j\in\mathcal{S}}(g(j)-g(i))P(i,j).
		\end{align*}
		
		The continuing proof directly follows the proofs of Proposition 8 and Proposition 10 in \cite{nelson2015b}.
	\end{proof}
	
	\begin{rmk}
		If we have a look at a state space containing only two states $\mathcal{S}=\{1,2\}$, we get
		\begin{align*}
		v(1)&=b(1)^2+P(1,2)(g(2)-g(1))^2-2b(1)P(1,2)(g(2)-g(1))\\
		v(2)&=b(2)^2+P(2,1)(g(1)-g(2))^2-2b(2)P(2,1)(g(1)-g(2))
		\end{align*}
		where $b(i)=a(i)-a^*$. This is the same result as the one derived in section \ref{twostates}, with $P(1,2)=p(1)$, $P(2,1)=p(2)$, $g(1)=g_1$ and $g(2)=g_2$.
	\end{rmk}
	
	\subsection{Implementation}
	
	The approaches used to calculate the diffusion coefficients for the model described above are similar to the ones explained in subsection \ref{implementation2states}.
	
	The question arises how to choose the values of $a(i)$ for $i \in \mathcal{S}$. Our quantile-based approach is illustrated in this section.
	
	Having calculated the price jumps, we split the data into two parts: One containing all the negative price changes and one containing the positive price jumps. Then evenly distributed quantiles are calculated for both sets of data. Depending on the data there might occur equal values for the quantiles. In this case we decrease the number of states as long as necessary. The state values $a(i)$ are set in the following way: We calculate the average of price changes located in between two quantiles or respectively below the first quantile and above the last one. 
	
	The single price changes are assigned in the following way to the states:
	\begin{itemize}
		\item price changes smaller than the smallest quantile are assigned to state 1
		\item price changes between the i\textsuperscript{th} and j\textsuperscript{th} quantile are assigned to state j
	\end{itemize}
	
	The approach is illustrated in figure \ref{Illustration}.
	\begin{figure}[t]
		\centerline{\includegraphics[width=0.75\textwidth]{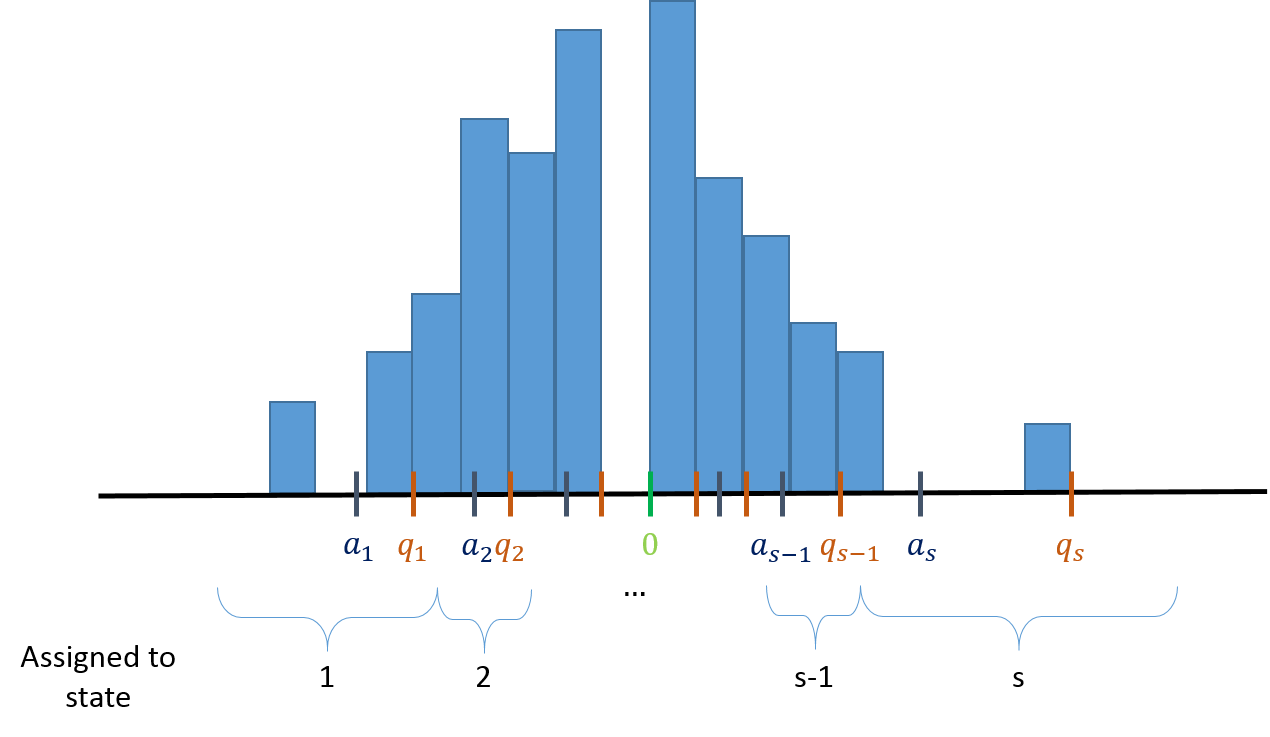}}
		\vspace*{8pt}
		\caption{Illustration}
		\label{Illustration}
	\end{figure}

	Based on these assignments, we follow the same approach as in subsection \ref{implementation2states} to calculate the transition matrix and do the remaining calculations.
	
	\subsection{Numerical Results}\label{Numerical_Results}
	The following table depicts the results gained from our computations for the \cite{lobster} data of Apple, Amazon, Google, Intel and Microsoft on 2012/06/21. The number of states we used to get the data, the matrix $P$ and the states $a(i)$ can be seen in figure \ref{PAmazon}. These two stocks serve here as examples and we can provide the same figures for Apple, Intel and Microsoft, if requested. As in subsection \ref{Num_Res} we do not consider the first and last 15 minutes of trading.
	\begin{table}[h]\centering
		\begin{tabular}{lccccc} 
			& $\sigma^2$ & $\tau^*$ & $m_\tau$ & $\dfrac{\sigma}{\sqrt{\tau^*}}$ & $\dfrac{\sigma}{\sqrt{m_\tau}}$\\ \hline
			Apple & 0.00031 & 0.0370 & \jump0.4026 & 0.0915 & 0.0277   \\ 
			Amazon & 0.00017 & 0.0892 & \jump0.9001 & 0.0433 & 0.0136 \\ 
			Google & 0.00090 & 0.1145 & \jump1.1291 & 0.0885 & 0.0282   \\ 
			Intel & 0.00004 & 1.3151 & 10.0839 & 0.0052 & 0.0019   \\ 
			Microsoft & 0.00004 & 0.8944 & \jump7.1647 & 0.0065 & 0.0023  \\ \hline
		\end{tabular}
		\caption{Numerical Results for many states, Mid Prices}
	\end{table}
	
	\begin{figure*} [b!]
		\subfigure{\includegraphics[width=0.49\textwidth]{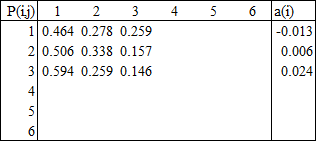}} 
		\subfigure{\includegraphics[width=0.49\textwidth]{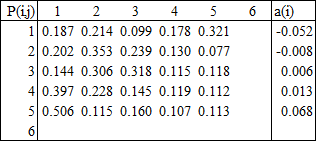}} 
		\vspace*{8pt}
		\caption{Amazon (left) and Google (right)}\label{PAmazon}
	\end{figure*} 
	
	We can show the linear relationship between the standard deviation of the ten minute price changes and the calculated diffusion coefficients, like we did in subsection \ref{Num_Res}. We calculate these standard deviations for our Apple, Amazon, Google, Intel and Microsoft data (provided in \cite{lobster}) and the data which is provided with \cite{Book}. We get an adjusted R\textsuperscript{2} of 0.9814 for the linear regression of the diffusion coefficients on the standard deviation for the balanced case and one of 0.9839 for the unbalanced case. The difference to the fit in \ref{Num_Res} is very small, which can be explained by having 25 of 40 stocks for which the algorithm sets only two states. In both cases the regression coefficients are highly significant.
	
	\begin{figure} [t]
		\subfigure{\includegraphics[width=0.49\textwidth]{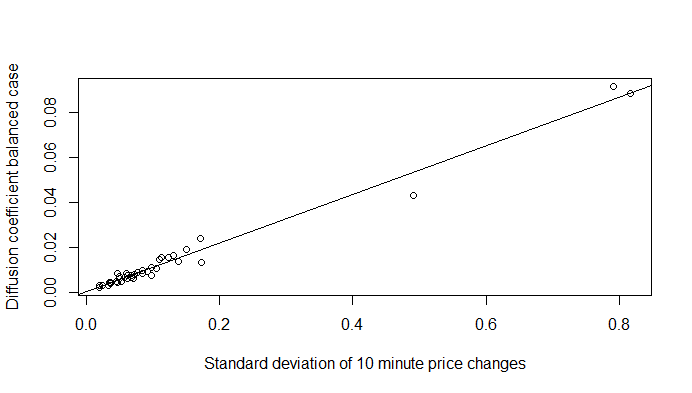}} 
		\subfigure{\includegraphics[width=0.49\textwidth]{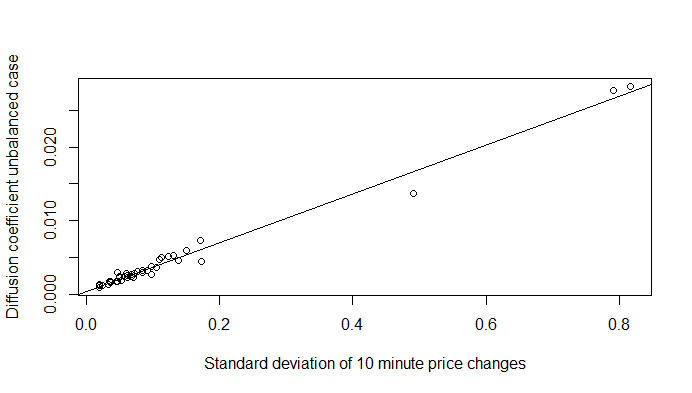}}
		\vspace*{8pt}
		\caption{Linear relationship of coefficients and standard deviation of 10 minutes mid-price changes for many states, Balanced case (left) and unbalanced case (right)} 
	\end{figure} 
	
	\section{Discussion on price spreads}
	
	A simplifying assumption used in \cite{ContLarrad} and \cite{nelson2015b} is that the spread, $p_t^a - p_t^b$ is fixed at a single tick, $\delta$. They justify this assumption by observing that over 98\% of all data points in their sample have a spread of $\delta$. Furthermore, they observe that the average lifetime, in ms, of a spread larger than a single tick is extremely small on average, often less than a single ms.\\
	In reproducing their figures, using our data, we found that over 90\% of all observations for AAPL, AMZN, and GOOG on 2012/06/21 has a spread stricly greater than $5\delta$. Table \ref{tickspread} displays the percentage of observations in some of our data sets at various tick sizes.
	
	\setlength{\tabcolsep}{1mm}
	\begin{table}[h]\centering
		\begin{tabular}{lccccccc} 
			Spread & 1 Tick & 2 Ticks & 3 Ticks & 4 Ticks & 5 Ticks&  $\geq 5$ Ticks & Avrg. Spread\\
			\hline
			AAPL & \jump0.79 & \jump1.79 & 2.10 & 2.44 &2.81 & 90.09 & 15.50\\
			AMZN & \jump1.31 &\jump1.52 & 1.74 & 2.23 & 2.78 & 90.43 & 13.59\\
			GOOG & \jump0.24 & \jump0.37 & 0.35 & 0.45 & 0.63 & 97.96 & 31.11\\
			INTC & 66.82 & 33.14 & 0.04 & 0.00 & 0.00 & 0.00 & \jump1.33\\
			MSFT & 65.63 & 34.31 & 0.05 & 0.00 & 0.00 & 0.00 & \jump1.34\\
			\hline
		\end{tabular}
		\caption{Percentage of Observations at various tick sizes for MSFT and INTC 2012/06/21 \label{tickspread}}
	\end{table}
	\setlength{\tabcolsep}{2mm}
	
	Even for our MSFT and INTC data the assumption that $p_t^a-p_t^b= \delta$ seems less warranted, given that a significant proportion of our observed data includes spreads greater than $\delta$. Figure \ref{histogramspread} displays histograms, for MSFT only, of the lifetimes of a spread greater than one tick and the lifetimes of a spread equal to one tick.
	\begin{figure}[t]
		\centerline{\includegraphics[width=\textwidth]{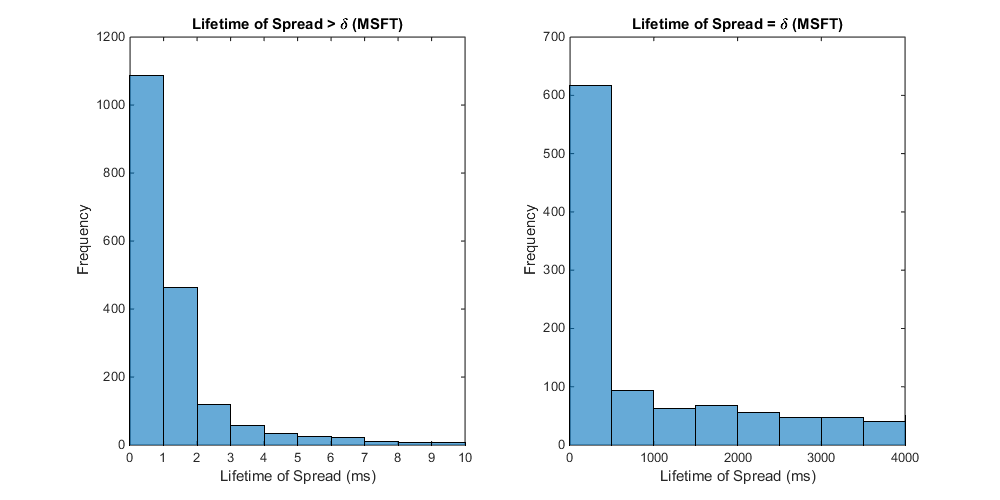}}
		\vspace*{8pt}
		\caption{Histograms of (left) lifetime of spread greater than one tick and (right) lifetime of spread equal to one tick for MSFT, 2012/06/21}
		\label{histogramspread}
	\end{figure}
	Further calculations demonstrate that, for MSFT and INTC, roughly 85\% of spreads greater than one tick had a lifetime less than 5 ms. While the assumption that spreads greater than $\delta$ are instantly filled do not hold as well for our data set, given the longer time scales used for asymptotic analysis we can justify this assumption for a subset of our sample.
	
	We believe that the larger spreads observed in our empirical data has consequences for estimated intensities of order flow. If we follow the method of \cite{ContStoikov} to estimate the rates of incoming Limit orders ($\lambda$), Market orders ($\mu$) and Cancellations ($\theta$), we only use the data points where $p_t^a-p_t^b = \delta$. The resulting estimations are included in table \ref{estTable}.
	
	For AAPL, AMZN, and GOOG, this means we use a small sample size to estimate these parameters; decreasing the accuracy of our estimations. Furthermore, since the vast majority of incoming orders occur at the best bid and ask (see \cite{ContLarrad}), and the best bid and ask are rarely $\delta$ apart, this means that we will have a much lower flow for these stocks empirically than observed in data where $p_t^a-p_t^b=\delta$ more often. For a subset of our sample, where the spread is much more frequently $\delta$, we find higher observed order intensities. However, even with the increased intensity of incoming order flow, we see that $\lambda > \mu + \theta$, which violates an important assumption of \cite{ContLarrad}. We believe that the increased instances of $p_t^a-p_t^b > \delta$ are also causing this violation. MSFT and INTC both have roughly 30\% of observed data points with spread $2\delta$. When the spread is greater than $\delta$, there is an incentive for traders to post limit orders within the spread so that their orders are executed first. We believe the resulting limit orders account for the increased estimated limit order intensity$\hat{\lambda}$.
	\begin{table}[H]\centering
		\begin{tabular}{lcc} 
			& $\hat{\lambda}$ & $\hat{\mu}+\hat{\theta}$ \\
			\hline
			MSFT & 3000.78 & 3000.72 \\
			INTC & 2483.13 & 2427.63\\
			GOOG & \jump \jump\jump 0.61 & \jump\jump \jump0.20\\
			AMZN & \jump \jump\jump3.76 & \jump\jump \jump1.15\\
			AAPL & \jump \jump\jump3.09 & \jump\jump \jump1.59\\
			\hline
		\end{tabular}
		\caption{Estimates for intensity of limit orders, market orders + cancellations in number of shares per second on 2012/06/21 \label{estTable}}
	\end{table}
	
	Given that our model focuses explicitly on the midprice, which does not include information about the observed spread, we assume $p_t^a- p_t^b = \delta$ for our model even though it is difficult to justify empirically using our data sets. As demonstrated, the diffusion limit calculated using the midprice is still a valid approximation for long run volatility of the midprice.
	
	\section{Conclusion and future work}
	
	After reviewing some of the assumptions from \cite{nelson2015b} and \cite{ContLarrad} for our sets of data, we showed how to develop the model presented in \cite{nelson2016b} (see sec. 6: Discussion). In section \ref{twostates}, we illustrated a model that considers two possible price changes different from one tick, as well as our numerical results. In section \ref{severalstates}, we further generalized the model by now considering an arbitrary number of possible price changes.
	
	As proposed in \cite{ContLarrad}, we compared the diffusion coefficients to the standard deviation of the ten minute price changes. Applying a linear regression using our available data, we tested how good our model extensions describe the linear relationship. We showed a large improvement in the adjusted R\textsuperscript{2} with the first extension, where we allowed sizes of price jumps to have a magnitude different to one tick. The second model extension included a higher number of possible sizes of price changes. The adjusted R\textsuperscript{2} increased again.
	
	Future work could take a closer look at  the model behaviour in times of crisis. It would be interesting to see if the distribution of inter-arrival times of book events changes and if a diffusion limit still exists. If evidence is found that regime switching has a big influence on the model, it would be the next step to include another Markov chain into the model containing the information in what kind of regime the market is. This extended model could be of use for practical algorithms trying to detect a switching point in price changes and make profit of that knowledge. It could also be used to detect insider trading. 
	
	Another extension of this model could challenge the assumption that the spread is fixed at $\delta$; allowing the spread to vary over time.
	\flushbottom
	\pagebreak

\end{document}